\documentclass[reqno]{amsart}
\usepackage{amsmath,amssymb,cite}
\usepackage[mathscr]{euscript}
\usepackage{graphics,epsfig,color}

\renewcommand{\epsilon}{\varepsilon}
\newcommand{\eps}{\varepsilon}

\newcommand{\rmd}{\mathrm{d}}

\theoremstyle{plain}
\newtheorem{theorem}{Theorem}[section]
\newtheorem{proposition}[theorem]{Proposition}
\newtheorem{lemma}[theorem]{Lemma}

\theoremstyle{definition}

\theoremstyle{remark}

\numberwithin{equation}{section}
\numberwithin{theorem}{section}

\title[ Concentrated vortex rings]{Time evolution of vortex rings with large radius
and very concentrated vorticity}

\author{Guido Cavallaro and Carlo Marchioro}

\begin{document}

\begin{abstract} 
We study the time evolution of an incompressible fluid with axial symmetry without swirl when the vorticity is sharply concentrated on $N$ annuli of radii $\approx$ $r_0$ and thickness $\epsilon$. We prove that when $r_0= |\log \epsilon|^\alpha, \,\, \alpha>2$, the vorticity field of the fluid converges as $\epsilon \to 0$ to the point vortex model, at least for a small but positive time.
This result generalizes a previous paper that assumed a power law for the relation between $r_0$ and $\epsilon$.
\end{abstract}



\maketitle

\thispagestyle{empty}

\section{Introduction} \label{sec1}

In the present paper we study the motion of an incompressible inviscid fluid with an axial symmetry  without swirl (for the exact definition see later on) when the initial vorticity is very concentrated on $N$ annuli of radii $\approx$ $r_0$ and thickness $\epsilon$. We prove the relation of this motion with the so-called point vortex system in the plane when $r_0\to \infty$ as $\epsilon \to 0$.

\medskip

\noindent The motion of an incompressible inviscid fluid is governed by the Euler equations, that for a fluid of unitary density in three dimensions read:

\begin{equation}
\label{1.1}
(\partial_t + (u\cdot \nabla))\omega=(\omega\cdot \nabla )u \, ,
\end{equation}
\begin{equation}
\label{1.2}
\nabla \cdot u =0   \qquad {\textnormal{(continuity  equation)}} ,
\end{equation}
$u(x,0)=u_0(x)$ (velocity field),  and boundary conditions. From now on we suppose that the velocity vanishes as $|x| \to \infty$. This assumption allows to reconstruct the velocity from the vorticity:
\begin{equation}
\label{1.3}
u(x,t)= - \frac{1}{4\pi} \int  \frac{x-y }{|x-y|^3} \wedge \omega(y,t) \, dy \, .
\end{equation}

\noindent We use now cylindrical coordinates $(z,r,\theta)$ and  suppose that the initial velocity field has the form (axial symmetry without swirl):
\begin{equation}
\label{1.4}
u(x,t)= (u_z,u_r,u_\theta)= (u_z(z,r,t), u_r(z,r,t), 0) \, .
\end{equation}
The time evolution conserves this symmetry. In this case the  vorticity is 
\begin{equation}
\label{1.5}
\omega= \nabla \wedge u = (0,0,\omega_{\theta}) = (0,0, \partial_zu_r - \partial_r u_z) \, ,
\end{equation}
and the  Euler equations become
\begin{equation}
\label{1.6}
(\partial_t  +(u_z\partial_z + u_r\partial_r)) \omega_{\theta} -\frac{u_r \omega_{\theta}}{r}  =0  \, ,
\end{equation}

\begin{equation}
\label{1.7}
\partial_z(r \ u_z)+\partial_r(r \ u_r)=0 \, .
\end{equation}

\noindent From now on we denote $\omega_{\theta}$ by $\omega$.
Finally, by (\ref{1.3}),
\begin{equation}
\label{1.8}
\begin{aligned}
u_z(z,r,t)= - \frac{1}{2\pi} & \int _{-\infty}^\infty dz'  \int_0^\infty r' \, dr'  \\
&\int _0^\pi d\theta \, \frac{\omega(z',r',t) [r \cos \theta - r']}{[(z-z')^2+(r-r')^2 + 2rr'(1-\cos \theta)]^{3/2}} \, ,
\end{aligned}
\end{equation}

\begin{equation}
\label{1.9}
\begin{aligned}
u_r(z,r,t)= \frac{1}{2\pi} & \int _{-\infty}^\infty dz'  \int_0^\infty r' \, dr'\\
 &\int _0^\pi d\theta\, \frac{\omega(z',r',t) [z - z']\cos \theta}{[(z-z')^2+(r-r')^2 + 2rr'(1-\cos \theta)]^{3/2}}  \, .
\end{aligned}
\end{equation}

Hence, the axially symmetric solutions to the Euler equations are given by the solutions to eqs. \eqref{1.5}-\eqref{1.9}.
Eq. (\ref{1.6}) means that the quantity $\omega/r$ remains constant along the flow generated by the velocity field, i.e.
\begin{equation}
\label{1.10}
\frac{\omega(z(t),r(t),t)}{r(t)}= \frac{\omega(z(0),r(0),t)}{r(0)} \, ,
\end{equation}
where $(z(t),r(t))$ solve
\begin{equation}
\label{1.11}
\dot{z}(t)=u_z(z(t),r(t),t) \, ,  \qquad  \dot{r}(t)=u_r(z(t),r(t),t) \, .
\end{equation}
It is possible to introduce an equivalent weak formulation of \eqref{1.6} that allows to consider non-smooth initial data;
by a formal integration by parts we obtain indeed
\begin{equation}
\label{1.12}
\frac{\mathrm{d}}{\mathrm{d}t} \omega_t(f)=\omega_t \big[u_z\partial_z f + u_r\partial_r f + \partial_tf \big] \, ,
\end{equation}
where $f=f(z,r,t)$ is a bounded smooth test function and
\begin{equation}
\label{1.13}
\omega_t (f):=\int_{-\infty}^\infty dz\int_0^\infty dr \ \omega (z,r,t) f(z,r,t) \, .
\end{equation}
It is known that a global (in time) existence and uniqueness of a weak solution to the associate Cauchy problem holds when the initial vorticity is a bounded function with compact support contained in the open half-plane  $\Pi:=\lbrace(z,r):  r>0\rbrace$, see  for instance  \cite{mpulv} pag.91 or the Appendix of \cite{CS}.  In particular, it can be shown that the support of the vorticity remains in the open half-plane $\Pi$ at any time.
A point  in the half-plane $\Pi$ denotes a circumference in the whole space.
The special class of axisymmetric without swirl solutions are called sometimes {\textit{smoke  rings}}, because there exist particular solutions whose shape remains constant in time (the so-called steady vortex ring) and translate in the $z$-direction with constant speed  (see for instance \cite{Fra70}). The existence and the properties of these solutions is an old question. For a rigorous proof by means of variational methods see  \cite{AmS89,FrB74}. For references on axially symmetric solution without swirl see also the review paper \cite{ShL92}.

\bigskip

Denote $x=(x_1, x_2) := (z, r-r_0)$.
 We assume that initially the vorticity is concentrated in $N$ blobs of the form
\begin{equation}
\omega_\epsilon(x,0)= \sum_{i=1}^N \omega_{i,\epsilon}(x,0)
\label{in_data}
\end{equation}
where $\omega_{i,\epsilon}(x,0)$ are functions with a definite sign such that, denoting by $\Sigma(\xi| \rho)$ the open disk of center $\xi$ and radius $\rho$ in $\mathbb{R}^2$,
\begin{equation}
\Lambda_{i,\epsilon}:= \text{supp} \, \omega_{i,\epsilon}(\cdot,0) \subset \Sigma(z_i | \epsilon) 
\; ; \qquad  \Sigma(z_i |\epsilon) \cap  \Sigma(z_j |\epsilon) =\emptyset \ \ \forall i \ne j \, 
\label{in_data2}
\end{equation}
being $\epsilon>0$
a small parameter and 
$z_1, \dots, z_N$, points contained in a bounded region of $\mathbb{R}^2$ such that
$$
\min_{i\neq j} |z_i -z_j| > \rho_m
$$
for a positive constant $\rho_m$ independent of $\epsilon$.
Moreover we assume that, for any $i=1, \dots, N$,
\begin{equation}
\int  dx \ \omega_{i,\epsilon}(x,0) := a_i \in \mathbb{R}
\end{equation}
independent of $\epsilon$ and 
\begin{equation}\label{mass}
|\omega_{i,\epsilon}(x,0)| \leq M \epsilon^{- \gamma} \; , \qquad M>0 \; , \quad \gamma >0 \, .
\end{equation}

\noindent We will discuss if,  in some cases and for small $\epsilon$,  the
time evolution of these states has the same form. In
Theorem \ref{teoEA}  we will prove that, as the assumption $r_0=|\log
\epsilon|^\alpha$   ($\alpha > 2$)  is fulfilled, the evolved state
$\omega_\epsilon(x,t)$ can be written as
\begin{equation}
\omega_\epsilon(x,t)= \sum_{i=1}^N \omega_{i,\epsilon}(x,t) \, ,
\label{t_data}
\end{equation}
where $\omega_{i,\epsilon}(x,t)$ are functions with definite sign such that
\begin{equation}
\Lambda_{i,\epsilon}(t):= \text{supp} \, \omega_{i,\epsilon}(\cdot, t) \subset \Sigma(z_i(t) | r_t(\epsilon)) \, ,
\end{equation}
with
\begin{equation}
 \Sigma(z_i(t) | r_t(\epsilon)) \cap  \Sigma(z_j(t) | r_t(\epsilon)) = \emptyset   \qquad \forall\, i \ne j \, ,
\end{equation}
being $r_t(\epsilon)$ a positive function, vanishing for  $\epsilon\rightarrow 0$, and
$z_i(t)\in \mathbb{R}^2$ solution to
the point-vortex model, that  is the dynamical system defined by the following differential equations:
\begin{equation} \label{pvmodel}
\dot{z}_i(t)= \sum_{\substack{j=1 \\j \neq i}}^N a_j K(z_i(t)- z_j(t)) \, , \qquad z_i(0)=z_i
\end{equation}
\noindent for $i=1, \ldots, N$,  and 
\begin{equation}\label{nucleoK}
K(x)=-\frac{1}{2\pi}\nabla^{\perp}\log|x|, \qquad \nabla^{\perp}=(\partial_2,-\partial_1),
\end{equation}

\noindent where $-1/2\pi\log|x|$ is the fundamental solution of the Laplace operator in $\mathbb{R}^2$. 
When all the $a_i$ have the same sign there is a global solution, otherwise there is a finite time
at which
a collapse (that is two $z_i$ arriving at the same point) or a $z_i$ going to infinity can
occur. 
However  (see for instance \cite{mpulv}) these events are exceptional.

\noindent Few words on this dynamical system: it has been introduced by Helmholtz as particular {\textit{solution}} of the Euler equations  \cite {Hel67} and investigated by many authors
\cite{Kel,Kir,Poi}.
It has been used to investigate the time evolution of irregular initial data and it produces an approximation method (called vortex method) in which $N \to \infty$ and $a_i\to 0$ (for more information see for instance the textbook  \cite{mpulv} or \cite{MaB02} and references in \cite{CGP14}).

Even if solutions of \eqref{pvmodel} cannot be a solution of the Euler equations, they can be an average of different solutions that in $\mathbb{R}^3$ are clusters of straight lines of vorticity, as it is discussed in \cite{CapMar, Mar88,  Mar,  MaP, MaP93, Turk}.  In \cite{marc} it is shown that the same happens when the straight lines are changed into large enough annuli,
with radius of the order $r_0(\epsilon)= \epsilon^{-\alpha}$,   for any $\alpha>0$ and for any finite time, and recently for long times in \cite{CS}.

In the present paper we assume a weaker dependence of $r_0$ on $\epsilon$ and we will show that the relation with the point vortex model remains valid at least for a finite but positive time.
We consider radii of the order $|\log\epsilon|^\alpha$, $\alpha>2$ (this lower bound on $\alpha$ 
  appears for a technical
reason which occurs in eq.  \eqref{qui}),
hence the rings are less distant from the axis, where the effects of curvature become stronger. 
For laws $r_0=f(\epsilon)>\epsilon^{-\alpha}$ the convergence to the point-vortex model
happens faster, while for the previous logarithmic law  the convergence is slower.
 This imposes a more careful strategy to prove the convergence, since we need an iterative method  as the one used in \cite{butmar2}
in a different context, which produces a result only for bounded times. This reflects the 
difficulty of the convergence when the distance from the axis approaches the scale
$|\log\epsilon|^\alpha$, $0\leq \alpha \leq 1$, for which only for $\alpha=0$ 
(for one vortex alone \cite{ben} and $N$ vortices \cite{butmar2}) and for $\alpha=1$  
for one vortex alone \cite{marneg} some results are available. 
For $\alpha=0$ the dynamics of the vortices converges to
simple translations parallel to the symmetry
axis with constant speed.
For $\alpha = 1$
we can conjecture
that  the convergence
of the dynamics is not to the point-vortex model, but to the dynamical
system defined by
\begin{equation}
\dot{z}_i(t)= \sum_{\substack{j=1 \\j \neq i}}^N a_j K(z_i(t)- z_j(t))
+ e_1 a_i \, , \qquad z_i(0)=z_i \, , \quad e_1=(1,0) \, .
\label{dyn2}
\end{equation}
  This is proved rigorously for one vortex alone in \cite{marneg}, but for
	$N$ vortices it is an open problem and the correspondence
	with \eqref{dyn2}  is established only at a heuristic level.

We also mention that in literature it is discussed how the point-vortex model behaves under a viscosity perturbation \cite{BrM11,CS, Gal11,Lad68, Mar90, marcNS, Mar07, UkY68}, but this topic is out of the scope of the present analysis.

\section{Main result} \label{sect2}

{\textit{A warning on the notation}}. Hereafter  in the paper
we denote by $C$
a generic positive constant (eventually changing from line to line) which is independent of the parameter $\epsilon$
and the time $t$.

We define a suitable scaling of variables, in order to get the convergence to the point-vortex dynamics, in such a way that  the rings increase their radius while their support becomes smaller. When the radius
increases, the interaction of the $N$ vortices with the axis becomes negligible in the limit $\epsilon\to 0$, which permits
to obtain a convergence to the point-vortex dynamics. Denoting by $(r, z, \theta)$  the cylindrical coordinates in $\mathbb{R}^3$, we recall the previously mentioned coordinates
\begin{equation} \label{coordinate} x=(x_1, x_2) := (z, r-r_0) \,  \end{equation}
and consider an initial vorticity as specified in \eqref{in_data}-\eqref{mass},
whose evolution at time $t$ can be expressed as in \eqref{t_data}.
Knowing the velocity field $u(\cdot,t):=u_t(\cdot)$, we can define the trajectory of a fluid element starting at $x$ as the solution of the integral equation
$$
\phi_t(x) = x + \int_0^t u_s (\phi_s(x)) ds \, .
$$
Since the quantity $\omega/r$ remains constant along the 
flow generated by the velocity field, we have
\begin{equation}\label{term i EA}
 \omega_{i, \epsilon}(x,t) :=  \frac{r_0+x_2}{r_0+(\phi_{-t}(x))_2} \omega_{i, \epsilon}(\phi_{-t}(x),0)  \, .
 \end{equation}

\noindent Moreover $\omega_{i,\varepsilon}(x,t)$ preserves the initial sign and the total mass $a_i$, as immediately follows from the definitions (see also Lemma \ref{norme di omega}).

\noindent Furthermore, for each index $i$, we can decompose the velocity field $u$ as follows:
\[u(x,t) =u^i(x,t) + F_\epsilon^i(x,t) \, , \]

\noindent where
$$
u^i(x,t)=\int dy\,G(x,y)\,\omega_{i,\varepsilon}(y,t)
$$ 
 
\noindent is the velocity field generated by the vortex $\omega_{i, \epsilon}$, and 
\begin{equation}
F_\epsilon^i(x,t)=\sum\limits_{j\neq i}\int dy\,G(x,y)\,\omega_{j,\varepsilon}(y,t)
\label{effe_iep}
\end{equation}

\noindent is the one generated by the remaining $N-1$ vortices; here $G(x,y)$ denotes the integral kernel appearing in \eqref{1.8}-\eqref{1.9} in the new coordinates \eqref{coordinate}.

Let us call $\{z_i(t) \}_{i=1, \ldots, N}$ the solution to point vortex dynamics \eqref{pvmodel} with intensities $a_i$ and initial data $z_i$ (with $z_i\neq z_j$  for  $i\neq j$). 
As already discussed, such dynamics is well defined globally in time apart from a zero
measure set of initial data. Even in this last case in which a collapse can occur,
since our results hold for times smaller than a positive constant, we can consider such
constant (let's call it $T_*$) much smaller than the first collapse time. With this viewpoint we define,
for a constant $\bar R>0$,
\begin{equation}
T_\omega := \sup \{ t>0 : \: \mathrm{supp}\,\omega_{i, \epsilon}(s) \subseteq \Sigma(z_i(s)| \bar R) \; \; \forall i=1, \ldots, N \,, \; \forall s \in [0,t]  \}  \, ,
\end{equation}
\begin{equation}
\bar T = \min \{ T_\omega,  T_*    \} \, ,
\end{equation}
and
\begin{equation}
 R_m := \min_{i \neq j} \inf_{t \in [0, \bar T)} |z_i(t) -z_j(t)| >0 \, . 
\end{equation}
We ask 
\begin{equation}
\bar R < R_m/4
\label{R_m}
\end{equation}
(it  will be used in the sequel). Observe that such requirement
is non-empty: for $\bar R=\epsilon$ it results $\bar T =0$ (by the initial data 
\eqref{in_data}-\eqref{in_data2}), and \eqref{R_m} is obviously fulfilled (for small $\epsilon$).
Considering then a small positive $\bar R$ (but independent of $\epsilon$), we obtain
consequently a small $\bar T$, and
 for a continuity
argument \eqref{R_m} can still be satisfied. In the next Theorem we state a better result,
the size of the support of $\omega_i$ at time $t$ (for short times) is a quantity which vanishes
for $\epsilon\to 0$.
The result is the following.
\begin{theorem} \label{teoEA}
Consider initial vorticity as in \eqref{in_data}-\eqref{mass} and $r_0 = |\log\epsilon|^{\alpha}$, for any $\alpha>2$.      
Then there is a $T>0$ such that
$$
\mathrm{supp}\,\omega_{i, \epsilon}(s) \subseteq \Sigma\left(z_i(s)  | C_T |\log\epsilon|^{-k}\right) \; \; \forall i=1, \ldots, N \,, \; \forall s \in [0,T]
$$
where $C_T$ is a positive constant,  $k = (\alpha-2)/2$,    $\epsilon \in(0, \epsilon_0)$, with 
$\epsilon_0<1$ solution to
$$
C_T |\log\epsilon_0|^{-k}=R_m/4 \, .
$$
\end{theorem}

\bigskip

\bigskip

\noindent We remark that with the previous definition of $\epsilon_0$ it results $C_T |\log\epsilon|^{-k}<R_m/4$.

\section{Proof of Theorem \ref{teoEA}}
We give the general strategy of the proof, which is rather technical and composed of 
many auxiliary Lemmas and Propositions.
We study the motion of a tagged vortex (with index $i$) under the influence of the remaining $N-1$. The field generated by the remaining $N-1$ vortices has the features of
a given external bounded field, since in the time interval $[0, \bar T]$ 
the minimum distance between any two distinct vortices
remains greater than a positive constant. We make use then of a fundamental estimate on the growth in time of the moment
of inertia of the tagged vortex, and we estimate the vorticity mass far from the center
of vorticity, showing that it is negligible when $\epsilon$ is small, by means of an
iterative method. Putting together
these (and other technical) results we achieve the proof.
Some of these tools are similar to those of 
previous papers \cite{butmar, butmar2, CapMar, CS, ISG, marc},  and  we
write them again for completeness.

\noindent We discuss then the preliminary results we need, starting with
the estimate of the convolution kernel $G$ (in the new coordinates \eqref{coordinate}), showing that, under suitable assumptions, this is near to $K$ (of the planar case).

\noindent 
Making use of
\eqref{1.8}-\eqref{1.9}, written with respect to the new coordinate system \eqref{coordinate},
 we get
\[ u(x,t)= \int_{\mathbb{R}^2}  G(x, y) \, \omega_\epsilon(y, t) \, dy \]
where the convolution kernel $G(x,y)$ is defined by:     
\begin{equation}
G_1(x,y) =  \frac{1}{2\pi} \int_0^{\pi} d\theta \, 
\frac{ (r_0+y_2) \big[(r_0+y_2) - (r_0+x_2) \cos \theta \, \big]}
{ \left\{ |x-y|^2 +2(r_0+x_2)(r_0+y_2)(1-\cos \theta) \right\}^{3/2}}
\end{equation}
\begin{equation}
G_2(x,y) =  \frac{1}{2\pi} \int_0^{\pi} d\theta \, \frac{(r_0+ y_2)(x_1-y_1) \cos \theta}{ \left\{|x-y|^2 +2(r_0+x_2)(r_0+y_2)(1-\cos \theta) \right\}^{3/2}} \, .
\end{equation}
We now want to give an estimate for this convolution kernel, in particular we want to show that, for small enough $\epsilon$, the vector field $u$ is near to the vector field
$\widetilde{u}$ corresponding to the planar case, namely
\[ \widetilde{u}(x, t) = \int_{\mathbb{R}^2} K(x-y) \, \omega_\epsilon(y, t) \, dy \,  \]
with $K$ defined in \eqref{nucleoK}.
We need the following lemma, whose proof is contained in \cite{CS}
and reported here for  completeness.

\begin{lemma}
Let us define, for $a>0$:
\[ I_1(a) := \int_0^\pi d\theta \frac{\cos\theta}{[a^2+2(1-\cos\theta)]^{3/2}} \, , \]
\[ I_2(a):= \int_0^\pi d\theta \frac{1 - \cos\theta}{[a^2+2(1-\cos\theta)]^{3/2}}  \, . \]
Denoting by $\chi_{(0,1)}(\cdot)$ the characteristic function of the interval $(0,1)$, the following equalities hold:
\begin{equation} 
I_1 (a) = a^{-2} + R_1(a) \, ,\qquad I_2 (a) = -\frac{1}{2} \log a \cdot \chi_{(0,1)}(a) + R_2(a) \, , 
\label{separ} 
\end{equation} 
where $a \cdot R_1(a)$ is bounded and $|R_2(a)| \leq C \min(1, \frac{1}{a})$.
\end{lemma}

\begin{proof} 
Let us consider first $I_2$. We recall that $1-\cos\, \theta = 2[\sin(\theta/2)]^2$ and we write the integral as:
\begin{equation} \label{int2} I_2= \int_0^\pi d\theta \frac{2 [\sin (\theta/2)]^2 \, \cos(\theta/2)}{ \{a^2+4 [\sin(\theta/2)]^2 \, \}^{3/2}}+ 
   \int_0^\pi d\theta \frac{2 [\sin (\theta/2)]^2 (1-\cos(\theta/2))}{ 
   \{a^2+4 [\sin(\theta/2)]^2 \, \}^{3/2}} \, . 
	\end{equation}  
By the substitution $z= 2 \sin (\theta/2)$,  for the first integral in the right hand side
of \eqref{int2} we have
$$
\begin{aligned}
\int_0^{2} dz \frac{z^2}{ 2[a^2+z^2]^{3/2}}  &= \frac{1}{2} \bigg[ \log( \sqrt{a^2+z^2}+ z) - \frac{z}{\sqrt{a^2+z^2}} \bigg]_{z=0}^{z=2}  \\
 &= -(a^2+4)^{-1/2} +\frac{1}{2} \log(2+\sqrt{a^2+4}) - \frac{1}{2} \log a \, .
 \end{aligned}
$$
We deduce that for $a \rightarrow 0$ this quantity is equal to $- \frac{1}{2} \log a $   plus a bounded rest, while for $a \rightarrow \infty$ it behaves like $a^{-1}$.
 
\noindent For the second integral in \eqref{int2}, first of all we have
\[ \int_0^\pi d\theta \frac{2 [\sin (\theta/2)]^2 (1-\cos(\theta/2))}{ 
   \{a^2+4 [\sin(\theta/2)]^2 \, \}^{3/2}} \leq \frac{1}{4} \int_0^\pi d\theta \, \frac{1-\cos(\theta/2)}{\sin(\theta/2)} \]
   which is a bounded integral; on the other hand
\[ \int_0^\pi d\theta \frac{2 [\sin (\theta/2)]^2 (1-\cos(\theta/2))}{ 
   \{a^2+4 [\sin(\theta/2)]^2 \, \}^{3/2}} \leq \frac{2}{a^3} \int_0^\pi d\theta \sin^2(\theta/2) (1- \cos(\theta/2)) \leq C a^{-3} . \]
	These estimates for the two integrals in \eqref{int2} show that  the equality for $I_2$
	in \eqref{separ}
	holds with
$R_2(a)$ bounded by a constant for small $a$, and 
   by $C a^{-1}$ for large $a$.
   
\noindent   We evaluate now $I_1$, first when $a<1$, by decomposing the integral as
\begin{equation} \label{int1} I_1= \int_0^\pi d\theta \frac{\cos(\theta/2)}{\{a^2+2[\sin(\theta/2)]^2 \,\}^{3/2}} +
  \int_0^\pi d\theta \frac{\cos \theta - \cos(\theta/2)}{ \{a^2+2[\sin(\theta/2)]^2\,\}^{3/2}} . \end{equation}
The first integral in the right hand side
of \eqref{int1} can be computed as before with the substitution $z= 2 \sin (\theta/2)$:
\[ \int_0^2 dz \frac{1}{[a^2+z^2]^{3/2}} =  \bigg[ \frac{z}{a^2 \sqrt{a^2+z^2}} \bigg]_{z=0}^{z=2} = \frac{2}{a^2\sqrt{a^2+4}} \]
which, for $a \rightarrow 0$, is equal to $a^{-2}$ plus a bounded rest.
The second integral in \eqref{int1} can be bounded by noticing that 
\[0 \leq \cos(\theta/2) - \cos\theta \leq 1-\cos\theta \qquad \text{ for }0 \leq \theta \leq \pi\]
and hence it can be bounded by $I_2$.

\noindent We analyse now the case $a\geq1$, observing that
\[ |I_1(a)| \leq a^{-3} \int_0^\pi d\theta \, |\cos \theta | = \frac{2}{a^3}\]
and so $|R_1(a)| \leq 2 a^{-3}+ a^{-2}$. 
In both cases, $a<1$ and $a\geq 1$, we have that $a \cdot R_1(a)$ is bounded, as it goes to zero like $a \log a$ when $a \rightarrow 0$, and  behaves like $a^{-1}$ when $a \rightarrow \infty$.
\end{proof}

\begin{proposition} \label{stimaD}
Consider $x,y$ such that:
\[   |x_2| \leq \frac{r_0}{2} \qquad |y_2| \leq \frac{r_0}{2}  \]
 and let $r_0 = |\log\epsilon|^{\alpha}$. Then, for $\epsilon$ small enough:
\begin{equation}
 |G(x,y)- K(x-y)| \leq \frac{C}{|\log\epsilon|^\alpha} \: \left( 1+  \log|\log\epsilon| + 
 \big|\log |x-y|\,\big| \cdot \chi_{(0,1)}(|x-y|) \right) \, . 
\label{prop3.2}
\end{equation}
\end{proposition}

\begin{proof}
Define $a:= |x-y|\: (r_0+x_2)^{-1/2} (r_0+y_2)^{-1/2}$.
\begin{equation}
\begin{aligned}
2\pi G_1(x, y) &= 
  \int_0^{\pi} d\theta \, \frac{(r_0+y_2)(y_2-x_2 \, \cos \theta +r_0 (1-\cos \theta))}
  {(r_0+x_2)^{3/2} \, (r_0+y_2)^{3/2} \,  \{a^2+2(1-\cos\theta)\}^{3/2} }  \\
&= \frac{y_2 \cdot \big(I_1(a)+I_2(a)\big) - x_2 \cdot I_1(a) + r_0 \cdot I_2(a) }{(r_0+x_2)^{3/2} (r_0+y_2)^{1/2}} \\
&= \frac{y_2-x_2}{(r_0+x_2)^{3/2} (r_0+y_2)^{1/2}} \cdot I_1(a) + \frac{(r_0+y_2)^{1/2}}{(r_0+x_2)^{3/2}} \cdot I_2(a) \\
& = \sqrt{ \frac{r_0+y_2}{r_0+x_2} } \cdot \frac{y_2-x_2}{|x-y|^2} + \frac{y_2-x_2}{(r_0+x_2)^{3/2} (r_0+y_2)^{1/2}} \cdot R_1(a)\, \\
& \qquad + \sqrt{\frac{r_0+y_2}{(r_0+x_2)^3}} \cdot I_2(a) \, .
\end{aligned}
\label{2piG}
\end{equation}
Note that $\frac{y_2-x_2}{|x-y|^2}$ is the first component of $2 \pi K(x-y)$, so we subtract this quantity and estimate $|G_1-K_1|$.

Let us put $A=\sqrt{\frac{r_0+y_2}{r_0+x_2}}$, hence we have
\[ |A-1|=\frac{|A^2-1|}{|A+1|}=\bigg(\frac{|y_2-x_2|}{r_0+x_2}\bigg)(1+A)^{-1} \leq \frac{2|x-y|}{r_0} \]
where in the last inequality we have used the assumption $x_2 \geq - \frac{r_0}{2} $. Furthermore
$$
\begin{aligned}
 \sqrt{\frac{r_0+y_2}{(r_0+x_2)^3}} = & \frac{1}{(r_0+x_2)^{1/2}\, (r_0+y_2)^{1/2}} \cdot \frac{r_0+y_2}{r_0+x_2}  \\
 \leq & \frac{1}{(r_0+x_2)^{1/2}\, (r_0+y_2)^{1/2}} \cdot \frac{(r_0+x_2) + |x-y|}{r_0+x_2}  \\
 \leq & \frac{1}{(r_0+x_2)^{1/2}\, (r_0+y_2)^{1/2}} + \frac{a}{r_0+x_2} \leq \frac{2}{r_0}(1+a) \, .
 \end{aligned}
$$
Collecting these estimates  into \eqref{2piG}, and calling $D:= G-K$, we get:
$$
\begin{aligned}
2\pi |D_1(x,y)| \leq & \frac{2|x-y|\,|y_2-x_2|}{r_0 \, |x-y|^2} + \frac{a \cdot R_1(a)}{r_0+x_2} + \frac{2}{r_0}(1+a) \cdot I_2(a)  \\
& \leq \frac{2}{r_0} + \frac{2 C}{r_0} + \frac{2}{r_0} \left(C - \frac{1}{2} \, \log a \cdot \chi_{(0,1)}(a) \right) .
\end{aligned}
$$
For $a \in (0,1)$
$$
\begin{aligned} 0 \leq -\log a= &  \log(r_0+x_2)^{1/2} + \log(r_0+y_2)^{1/2} - \log |x-y|  \\
 \leq & \log r_0 + \log (3/2) +
\big|\log|x-y| \big| \cdot \chi_{(0,1)} (|x-y|)
\end{aligned}
$$
where, by assumption, $\frac{1}{2} r_0 \leq r_0 + x_2 \leq \frac{3}{2}r_0$. 
Since $r_0= |\log\epsilon|^{\alpha}$  we obtain
\[ 2\pi |D_1(x,y)| \leq \frac{C}{|\log\epsilon|^\alpha}  \bigg[ 1 +  \log|\log \epsilon| +
\big|\log|x-y| \big| \cdot \chi_{(0,1)} (|x-y|) \bigg] .\]

\noindent We analyse now  the second component:
$$
\begin{aligned}
 2\pi G_2(x,y) = & \int_0^\pi \frac{(r_0+y_2)(x_1-y_1)\, \cos \theta}{(r_0+x_2)^{3/2} \, (r_0+y_2)^{3/2} \,  \{a^2+2(1-\cos\theta)\}^{3/2}} \\
 = & \frac{x_1-y_1}{(r_0+x_2)^{3/2} \, (r_0+y_2)^{1/2}} \cdot I_1(a) \\
  = & \sqrt{\frac{r_0+y_2}{r_0+x_2}} \cdot \frac{x_1-y_1}{|x-y|^2} + \frac{x_1-y_1}{(r_0+x_2)^{3/2} \, (r_0+y_2)^{1/2}} R_1(a) \, .
  \end{aligned}
$$
We proceed as before,
\[ 2 \pi |D_2(x,y)| \leq \frac{2|x-y|\,|x_1-y_1|}{r_0 \, |x-y|^2} + \frac{a \cdot R_1(a)}{r_0+x_2} \leq \frac{2}{r_0}(1+C) \, .\]
Since $|D| \leq |D_1|+|D_2|$, we get
\[ |D(x,y)| \leq \frac{C}{|\log\epsilon|^\alpha} \left(1+ \log|\log \epsilon| + \big|\log |x-y| \big| \cdot \chi_{(0,1)} (|x-y|) \right) \, ,\]
and the Proposition is thus proved.
\end{proof}
The next lemma states the conservation of the $L^1$ norm of $\omega$ and a bound on its $L^\infty$
norm.
\begin{lemma} \label{norme di omega}
Let $\omega_{i, \epsilon}$  as in  \eqref{in_data}-\eqref{t_data}. Then
\begin{equation} 
\label{norme}
 \int_{\mathbb{R}^2} dx \, \omega_{i, \epsilon}(x,t) = \int_{\mathbb{R}^2} dx \, 
\omega_{i, \epsilon}(x,0) = a_i 
\end{equation}
and for each time $\omega_{i, \epsilon}(x, t)$ has the same sign of $\omega_{i, \epsilon}(x,0)$. 
Moreover for $t \leq  \bar T $ and for small enough $\epsilon$, 
\begin{equation} 
\label{norme2}
 |\, \omega_{i, \epsilon}(x, t) | \leq 3 M\, \epsilon^{-\gamma} \, .
\end{equation}
\end{lemma}
\begin{proof}
Equation \eqref{norme} is a direct consequence of the conservation of $\omega/r$, after
integrating in $\mathbb{R}^3$ adopting cylindrical coordinates.
The conservation of sign is evident by the definition of $\omega_{i, \epsilon}(x,t)$.
To obtain \eqref{norme2}  we observe that, if $x \in \Lambda_{i,\epsilon}(0) := \mathrm{supp} \, \omega_{i, \epsilon}(0)$, then
$|x-z_i| \leq \epsilon$ and
\[ |r_0+x_2| = ||\log\epsilon|^\alpha + x_2 | \geq |\log\epsilon|^\alpha - |z_i|
- \epsilon \geq \frac{1}{2} |\log\epsilon|^\alpha \, . \]
Since $\phi_t(x) \in \Lambda_{i, \epsilon}(t)$, it results $|\phi_t(x)- z_i(t)| \leq \bar R$ and then
\[ |r_0+ \phi^t_2(x)| \leq |\log\epsilon|^\alpha + |z_i(t)|+ \bar R \, . \]
Moreover
\[ |z_i(t)| \leq |z_i|+ \frac{1}{2\pi}\int_0^t ds \, \sum_{j \neq i}   \frac{|a_j|}{|z_i(s)- z_j(s)|} 
\leq |z_i|+ \frac{t}{2 \pi R_m} \sum_{j \neq i} |a_j| \, , \]
therefore, for $t \leq \bar T$ and for $\epsilon$ small enough,
\[ |r_0+ \phi^t_2(x)| \leq \frac{3}{2} |\log\epsilon|^\alpha \, .\]
The bound \eqref{norme2} follows from the equality
\[ \omega_{i, \epsilon}(\phi_t(x),t)= \frac{r_0+ \phi^t_2(x)}{r_0+x_2} \omega_{i, \epsilon}(x,0) \, .\]
\end{proof}

\noindent We are now able to prove that the difference $u-\widetilde{u}$ is small.
\begin{proposition} \label{u e utilde}
Let $\omega_{\epsilon}$ as in  \eqref{in_data}-\eqref{t_data}, and let $t \leq \bar T$, $x \in \Lambda_{\epsilon}(t)$. 
Then, if $\epsilon$ is small enough, 
\begin{equation} 
|u(x,t)-\widetilde{u}(x, t)| \leq \frac{C}{|\log\epsilon|^{\alpha-1}} \, .
\label{u e utilde 2}
\end{equation}
\end{proposition}

\begin{proof}
Note that, if $x \in \Lambda_{i, \epsilon}(t)$ (for any $i$), then $|x_2| \leq r_0/2$, 
as seen in the proof of Lemma \ref{norme di omega}. 
We can then apply Proposition \ref{stimaD} to bound $|u(x,t)-\widetilde{u}(x, t)|$,
and the worst term to treat is
\begin{equation}
\begin{aligned}
 &\int_{|x-y| <1} dy \, \big|\log |x-y|\big| \,\, |\omega_\epsilon(y, t)|  \\
&\leq\int_{|x-y| <1} dy \, \big|\log |x-y|\big| \, \sum_{i=1}^N |\omega_{i, \epsilon}(y, t)| \, .
\end{aligned}
\label{arrang}
\end{equation}
We use a classical trick, that is a rearrangement:
we bound this integral with the one obtained by concentrating as much as possible the vorticity around the singularity of $\log |x-y| $, namely $y=x$
(and by the characteristic function we integrate only in the domain $|x-y|<1$).
Let us proceed with a fixed $i$ in the summation in \eqref{arrang}.
Since the integral of $\omega_{i, \epsilon}$ is constant in time,  
and its $L^\infty$ norm is less or equal to $3 M \epsilon^{-\gamma}$, 
we get the rearrangement replacing $\omega_{i, \epsilon}$ with the function equal to the constant
$3 M\epsilon^{-\gamma}$ in the disk of centre $x$ and radius $r$, and equal to zero outside this disk.
The radius $r$ is chosen such that the total mass of vorticity is $|a_i|$, 
so $\pi r^2 \cdot 3 M\epsilon^{-\gamma} =|a_i|$.
We have then:
$$
\begin{aligned}
\int_{|x-y| <1} dy \, \big|\log |x-y|\big| \,\, |\omega_{i,\epsilon}(y, t)| \leq & \: \: 3 M \epsilon^{-\gamma} \int_{ \{ |x-y| \leq r \} } d y \big|\log |x-y| \big| \\
=& - 6 \pi M \epsilon^{-\gamma} \int_0^r \rho \log(\rho) d\rho \\
=&  - 6 \pi M \epsilon^{-\gamma} \bigg[ \frac{\rho^2}{2} \log \rho - \frac{\rho^2}{4} \bigg]_{\rho=0}^{\rho=r} \\
=& - 3 \pi M \epsilon^{-\gamma} r^2 \log r + \frac{3}{2} \pi M \epsilon^{-\gamma} r^2 .
\end{aligned}
$$
Since $r=\sqrt{\frac{|a_i| \epsilon^\gamma}{ 3 \pi M}}$, we get
\[ \int_{|x-y| <1} dy \, \big|\log |x-y|\big| \, \sum_{i=1}^N |\omega_{i, \epsilon}(y, t)| \leq C+ C \, |\log \epsilon| , \]
which inserted in \eqref{prop3.2} gives \eqref{u e utilde 2}.
\end{proof}
We give here a useful split of the field $F^i_\epsilon$.
\begin{lemma} \label{comesonoFeps}
Recalling the definition of $F^i_\epsilon$ given in  \eqref{effe_iep}, we can write, for $t \leq \bar T$
\[ F^i_{\epsilon}= F^i_{\epsilon, 1} + F^i_{\epsilon, 2} \]
where $F^i_{\epsilon, 1}$ is Lipschitz and bounded uniformly in $\epsilon$, and $F^i_{\epsilon, 2}$ is small, i.e. 
\[\|F^i_{\epsilon, 2}\|_{L^\infty} \leq \frac{C}{|\log \epsilon|^{\alpha-1}} \, . \]
\end{lemma}
\begin{proof}
 Let us define
$$
 \begin{aligned}
  F^i_{\epsilon, 1}(x,t) = & \sum_{j \neq i} \int dy \, K(x-y) \, \omega_{j, \epsilon}(y,t)\, , \\
  F^i_{\epsilon, 2}(x,t) = & \sum_{j \neq i} \int dy \, [G(x,y) - K(x-y)] \, \omega_{j, \epsilon}(y,t)\, .
 \end{aligned}
$$
It results that $F^i_{\epsilon, 1}$ is Lipschitz and bounded uniformly in $\epsilon$ by
the properties of $K$ outside the disk $\Sigma(0|R_m/2)$  (it is Lipschitz 
and bounded); in fact, if $x \in \Lambda_{i, \epsilon}(t)$ and $y \in \Lambda_{j, \epsilon}(t)$, for $i \neq j$, we have
\[ |x-y| \geq |z_i(t)-z_j(t)|- |x-z_i(t)|-|y-z_j(t)| \geq R_m -2 \bar R \geq \frac{R_m}{2} \]
for $4 \bar R < R_m$. The smallness of $F^i_{\epsilon, 2}$ is achieved by Proposition \ref{u e utilde}.
\end{proof}

We define now the center of vorticity and the moment of inertia,
\begin{equation}
 \begin{aligned}
&B^i_\epsilon(t):= a_i^{-1} \, \int_{\mathbb{R}^2} dx \, x \, \omega_{i,\epsilon}(x,t)  \\
&I^i_\epsilon(t):= \int_{\mathbb{R}^2} dx \, |x-B^i_\epsilon(t)|^2 |\omega_{i, \epsilon}(x,t)| 
\end{aligned}
\label{cen_vort}
\end{equation}
whose properties will be exploited in the following.

In the next lemmas we omit for simplicity the index $i$ from the notation and we assume, without lost of generality, $a_i=1$. 
This is equivalent to consider a ``reduced system" with only one vortex moving in an external field acting on it, which has the properties stated in Lemma \ref{comesonoFeps}. 
In this case it is easily  verified that the following equation holds:
\begin{equation} \label{eqderivaterid} \frac{\mathrm{d}}{\mathrm{d}t} \omega_t(f) = \omega_t \big[ ((u+F_\epsilon)\cdot \nabla) f + \partial_t f \big] \, . \end{equation}
The results we will prove hold obviously for each $i$.

\begin{lemma} \label{lem:Ieps Beps}
For $t \leq \bar T$ and for small enough $\epsilon$,
\begin{equation} \label{Iepsilon} I_\epsilon(t) \leq \frac{C}{|\log\epsilon|^{2(\alpha-1)}} \,.
 \end{equation}
\end{lemma}
\begin{proof}
We estimate the derivative of $I_\epsilon(t)$, using \eqref{eqderivaterid}:
$$
\begin{aligned}
 \frac{\mathrm{d}}{\mathrm{d}t} I_\epsilon(t) = & 
 \int dx \, \omega_\epsilon(x, t) \, \left[ (u+F_\epsilon)\cdot 2(x-B_\epsilon(t)) - \dot{B}_\epsilon(t) \cdot 2 (x-B_\epsilon(t)) \right] .
 \end{aligned}
$$
 Moreover
 \[ \frac{\mathrm{d}}{\mathrm{d}t} B_\epsilon(t) = \int dx \, \omega_\epsilon(x,t) \, (u(x,t)+F_\epsilon(x,t)) , \]
 then
$$
\begin{aligned}
 \frac{\mathrm{d}}{\mathrm{d}t} I_\epsilon(t) = & 
\,  2 \int dx \, \omega_\epsilon(x,t) \left[ u(x,t)-\int dy \, \omega_\epsilon(y,t) \, u(y,t) \right] \cdot (x-B_\epsilon(t))  \\
 & \: + 2 \int dx \, \omega_\epsilon(x,t) \left[ F_\epsilon(x,t)-\int dy \, \omega_\epsilon(y,t) \, F_\epsilon(y,t) \right] \cdot (x-B_\epsilon(t)) .
 \end{aligned}
$$
Consider first  the term containing $F_\epsilon$: we note that, by the definition of $B_\epsilon(t)$,
\[ \int dx\, \omega_\epsilon(x,t) (x-B_\epsilon(t)) \cdot \int dy \, \omega_\epsilon(y,t) F_\epsilon (y,t) =0 \]
\[ \int dx\, \omega_\epsilon(x,t) (x-B_\epsilon(t)) \cdot F_{ \epsilon,1}(B_\epsilon(t), t) =0 \, . \]
We then obtain:
$$
\begin{aligned} 
& 2 \left| \,\int dx \, \omega_\epsilon(x,t) \left[ F_\epsilon (x,t)-\int_{\mathbb{R}^2} dy 
\, \omega_\epsilon(y,t) \, F_\epsilon (y,t) \right] \cdot (x-B_\epsilon(t)) \, \right| \\
 & = 2 \left| \, \int dx \, \omega_\epsilon(x,t)
 \left[ F_{\epsilon, 1}(x,t)- F_{\epsilon, 1} (B_\epsilon(t),t) \right] 
 \cdot (x-B_\epsilon(t)) \, \right| \\
 & \quad + 2 \left| \int dx \, \omega_\epsilon(x,t) \, F_{\epsilon, 2}(x, t) \cdot (x-B_\epsilon(t)) \, \right| \\
 & \leq 2\int dx \, \omega_\epsilon(x,t)\,  L |x-B_\epsilon(t)|^2 + \frac{C}{|\log\epsilon|^{\alpha-1}} \int dx \, |x-B_\epsilon(t)|\, \omega_\epsilon(x,t)  \\
 & \leq 2L \, I_\epsilon(t) + \frac{C}{|\log\epsilon|^{\alpha-1}} [I_\epsilon(t)]^{1/2}
 \end{aligned}
$$
where, in the last line, we used Cauchy-Schwarz inequality, and  $L$ is the Lipschitz constant
of $F_{\epsilon, 1}$.

\noindent For  the term containing $u$, we have analogously:
 \[ \int dx\, \omega_\epsilon(x,t) (x-B_\epsilon(t)) \cdot \int dy \, \omega_\epsilon(y,t) u(y,t) =0 .\]
Moreover, by the  antisymmetry of $K$,
\begin{equation}
 \int dx \, \omega_\epsilon(x,t) \, \widetilde{u}(x,t) = \int dx \int dy \: \omega_\epsilon(x,t)\,\omega_\epsilon(y, t) \, K(x-y) = 0 
\label{antisym}
\end{equation}
and recalling that, by definition, $(x-y) \cdot K(x-y) =0$, we get
$$
\begin{aligned} \int dx \, \omega_\epsilon(x,t) \, x \cdot \widetilde{u}(x,t) = & 
\int dx \int dy \, \omega_\epsilon(x,t)\,\omega_\epsilon(y, t) \, x \cdot K(x-y) \\
=& \int dx \int dy \, \omega_\epsilon(x,t)\,\omega_\epsilon(y, t) \, y \cdot K(x-y) \end{aligned}
$$
hence this integral is zero as well, by the antisymmetry of $K$.
Using Proposition \ref{u e utilde} we get then:
$$
\begin{aligned}
& 2 \left| \,\int dx \, \omega_\epsilon(x,t) \left[ u(x,t)-\int dy \, \omega_\epsilon(y,t) \, u(y,t) \right] \cdot (x-B_\epsilon(t)) \, \right|  \\
& \leq \, 2 \int dx \, \omega_\epsilon(x,t) \, \big|u(x,t)-\widetilde{u}(x,t)\big| \, |x-B_\epsilon(t)|\\\
& \leq \:  \frac{C}{|\log\epsilon|^{\alpha -1}} \int dx \, \omega_\epsilon(x,t) 
\, |x-B_\epsilon(t)| \leq \: \frac{C}{|\log\epsilon|^{\alpha -1}} \,  \left[I_\epsilon(t) \right]^{1/2}
\end{aligned}
$$
where Cauchy-Schwarz inequality has been used again in the last line. 

Hence we have
\[ |\dot{I}_\epsilon(t)| \leq 2L\,I_\epsilon(t) + \frac{C}{|\log\epsilon|^{\alpha -1}} \left[I_\epsilon(t) \right]^{1/2} . \]
Defining $M_\epsilon(t):= [I_\epsilon(t)]^{1/2}$, by Gronwall's inequality and using the fact  that,
by the initial data,  $I_\epsilon(0) \leq 4\epsilon^2$, we get
\begin{equation} 
 M_\epsilon(t) \leq \left( 2\epsilon + 
\frac{C}{2 L|\log\epsilon|^{\alpha -1}} \right) {\textnormal{e}}^{Lt} 
\leq  \frac{C}{|\log\epsilon|^{\alpha -1}} \, {\textnormal{e}}^{Lt} \, . 
\label{sqrtIepsilon}
\end{equation}
From the previous bound we finally obtain, recalling that $t\leq \bar T$,
\[ I_\epsilon(t) \leq \frac{C}{|\log\epsilon|^{2(\alpha -1)}} \, . \]
\end{proof}

\begin{lemma} \label{variazionemax}
Let us put
\begin{equation} 
 R_t := \max \{ |x-B_\epsilon(t)|: \: x \in \Lambda_\epsilon(t) \} 
\label{max_lambda} 
\end{equation} 
and choose $x_0 \in \Lambda_\epsilon(0)$ such that, at time $t \leq \bar T$, 
\begin{equation} 
 \frac34 R_t\leq|\phi_t(x_0)- B_\epsilon(t)| \leq R_t \, . 
\label{erret} 
\end{equation} 
Then at this time $t$ the following inequality holds:
\begin{equation} \label{eq:variazonemax} \frac{\mathrm{d}}{\mathrm{d}t}|\phi_t(x_0)- B_\epsilon(t)| \leq 2LR_t + \frac{C I_\epsilon(t)}{ R_t^3}
 + \sqrt{\frac{3\, M \epsilon^{-\gamma} m_\epsilon(R_t/2, t)}{\pi} } + \frac{C}{|\log\epsilon|^{\alpha -1}}
\end{equation} 
where the function $m_\epsilon$ is defined by:
\begin{equation} 
 m_\epsilon(R, t) := \int_{|y-B_\epsilon(t)|>R} dy \, \omega_\epsilon(y,t) \qquad \text{for }R \in (0, +\infty) \, .
\label{funzemme} 
\end{equation} 
\end{lemma}

\begin{proof}
Let us put $x= \phi_t(x_0)$. We have:
$$
\begin{aligned} &\frac{\mathrm{d}}{\mathrm{d}t}|\phi_t(x_0)- B_\epsilon(t)| = 
[u(x,t)+F_\epsilon(x,t)- \dot{B}_\epsilon(t)] 
 \cdot \frac{x-B_\epsilon(t)}{|x-B_\epsilon(t)|} \\
 & \qquad = \left[ \int dy \, (F_\epsilon (x,t)-F_\epsilon (y,t)) \, \omega_\epsilon(y,t) \right] \cdot \frac{x-B_\epsilon(t)}{|x-B_\epsilon(t)|} \\
 & \qquad \quad + \left[ \int dy \, (u(x,t)-u(y,t)) \, \omega_\epsilon(y,t) \right] \cdot \frac{x-B_\epsilon(t)}{|x-B_\epsilon(t)|} .
\end{aligned}
$$
The term involving $F_\epsilon$ is easily bounded using Proposition \ref{comesonoFeps}, 
since $F_{\epsilon, 1}$ is Lipschitz and $F_{\epsilon, 2}$ is small:
\begin{equation}
\begin{aligned} \int dy \, |F(x,t)-F(y,t)| \, \omega_\epsilon(y,t) \leq & 
 L \int dy \, |x-y| \, \omega_\epsilon(y,t) +  \frac{C}{|\log\epsilon|^{\alpha -1}} \\
 \leq & 2L\, R_t + \frac{C}{|\log\epsilon|^{\alpha -1}} .
\end{aligned}
\label{f1}
\end{equation}
For the second integral containing the difference of the velocity field, we split it into three terms, 
recalling that \\ $\int dy\, \widetilde u(y)\, \omega_\epsilon(y) =0$:
\begin{equation}
 |u(x,t)- \widetilde{u}(x, t)| \leq \, \frac{C}{|\log\epsilon|^{\alpha -1}} , 
\label{f2}
\end{equation}
\begin{equation}
 \left|\int dy \, (u(y,t)-\widetilde{u}(y,t) )\, \omega_\epsilon(y,t) \right| \leq \,\frac{C}{|\log\epsilon|^{\alpha -1}} \, . 
\label{f3}
\end{equation}
The third (non trivial)  term is
\[ \widetilde{u}(x,t) \cdot \frac{x-B_\epsilon(t)}{|x-B_\epsilon(t)|} = 
\frac{x-B_\epsilon(t)}{|x-B_\epsilon(t)|} \cdot \int dy\, K(x-y) \omega_\epsilon(y,t) \, .\]
The integration domain can be decomposed into two regions: $A_1 := \Sigma\big(B_\epsilon(t) \big| R_t/2\big) $ and
$A_2:= \mathbb{R}^2 \setminus A_1\,$. We call $H_1$ and $H_2$ the resultant integrals. We follow
for $H_1$ the proof of \cite[Lemma~2.5]{butmar};
recalling \eqref{nucleoK} and the notation $x^\perp = (x_2,-x_1)$ for $x = (x_1,x_2)$, after introducing the new variables $x'=x-B_\epsilon (t)$,  $y'=y-B_\epsilon (t)$, and using that $x'\cdot (x'-y')^\perp=-x'\cdot y'^\perp$, we get,
\begin{equation}
\label{in H_11}
H_1 = \frac{1}{2\pi} \int_{|y'|\leq R_t/2}\! \rmd y'\, \frac{x'\cdot y'^\perp}{|x'||x'-y'|^2}\, \omega_\epsilon (y'+B_\epsilon (t))\;.
\end{equation}
By  definition of  center of vorticity \eqref{cen_vort}, $\int\! \rmd y'\,  y'^\perp\, \omega_\epsilon (y'+B_\epsilon (t)) = 0$, so that
\begin{equation}
\label{in H_13*}
H_1  = H_1'-H_1''\;, 
\end{equation}
where
\begin{eqnarray*}
&& H_1' = \frac{1}{2\pi}  \int_{|y'|\le R_t/2}\! \rmd y'\, \frac {x'\cdot y'^\perp}{|x'|}\, \frac {y'\cdot (2x'-y')}{|x'-y'|^2 \ |x'|^2} \, \omega_\epsilon (y'+B_\epsilon (t))\;, \\ && H_1''= \frac{1}{2\pi} \int_{|y'|> R_t/2}\! \rmd y'\, \frac{x'\cdot y'^\perp}{|x'|^3}\, \omega_\epsilon 
(y'+B_\epsilon (t))\;.
\end{eqnarray*}
From \eqref{erret} we have $|x'| \geq 3R_t/4 $, and hence $|y'| \le R_t/2$ implies $|x'-y'|\ge R_t/4$ and $|2x'-y'|\le |x'-y'|+|x'| \le  |x'-y'| +R_t \le 5|x'-y'|$, so that
\begin{equation*}
|H_1'|\leq \frac{C}{R_t^3}  \int_{|y'|\leq R_t/2} \! \rmd y'\, |y'|^2 \, \omega_\epsilon (y'+B_\epsilon (t)) \le \frac{C I_\epsilon (t)}{ R_t^3}\;.
\end{equation*}
To bound $H_1''$ we note that, in view of \eqref{max_lambda}, the integration is restricted to $|y'|\le R_t$, so that (using the lower bound in \eqref{erret})
\begin{equation*}
|H_1''| \le \frac{C}{R_t} \int_{|y'|> R_t/2}\! \rmd y'\, \omega_\epsilon (y'+B_\epsilon (t))\le \frac{C I_\epsilon (t)}{ R_t^3}\;,
\end{equation*}
where in the last inequality we used the Chebyshev's inequality. By \eqref{in H_13*} and the previous estimates we conclude that
\begin{equation}
\label{H_14}
|H_1| \le \frac{C I_\epsilon (t)}{ R_t^3}\;.
\end{equation}
We analyse now $H_2$. We first note that:
\[ |H_2| \leq \frac{1}{2\pi} \int_{|y-B_\epsilon(t)|>R_t/2} dy \, \frac{1}{|x-y|} \, \omega_\epsilon(y, t) \]
so we can bound $H_2$ using again rearrangement, as in the proof of Proposition \ref{u e utilde}; 
we bound the integral taking a vorticity concentrated, as much as possible, around the singularity of $\frac{1}{|x-y|}$. 
Therefore, the rearrangement is achieved defining $\omega_\varepsilon$ equal to
$3\,M \epsilon^{-\gamma}$ for $|x-y|<r$, and equal to $0$ for $|x-y| \geq r$, 
where $r$ is chosen such that $3 \, M \epsilon^{-\gamma} \cdot \pi r^2= m_\epsilon(R_t/2,t)$ 
(which is the ``total mass'' of $\omega_\epsilon$ in the integration domain $A_2$). Then
$$
\begin{aligned}
|H_2| \leq & \: \frac{3\, M \epsilon^{-\gamma}}{2 \pi} \int_{|z|<r} \frac{dz}{z} = 3\, M \epsilon^{-\gamma} \int_0^r \frac{\rho}{\rho} \, d\rho 
= 3\, M \epsilon^{-\gamma} \, r = \sqrt{\frac{3\,M m_t(R_t/2)}{\pi \, \epsilon^\gamma}} .
\end{aligned}
$$
Collecting this last estimate, \eqref{f1}, \eqref{f2}, \eqref{f3}, and \eqref{H_14},  we obtain \eqref{eq:variazonemax}.
\end{proof}

We investigate now the behavior near to $0$ of the function
$m_\epsilon (\cdot, t)$ introduced in \eqref{funzemme}.

\begin{lemma} \label{lem:mt}
There exists $T>0$ such that,
for each $\ell>0$,
\begin{equation}
 \lim_{\epsilon \rightarrow 0} \epsilon^{-\ell}m_\epsilon\left(\frac{1}{|\log\epsilon|^k}, t\right) =0  
\label{lemma_massa}
\end{equation}
for any $t\in [0, T]$ and $k =(\alpha-2)/2$.
\end{lemma}

\begin{proof}
Given $R\ge 2h>0$, let $W_{R,h}(x)$, $x\in \mathbb{R}^2$, be a non-negative smooth function, depending only on $|x|$, such that
\begin{equation}
\label{W1}
W_{R,h}(x) = \begin{cases} 1 & \text{if $|x|\le R$}, \\ 0 & \text{if $|x|\ge R+h$}, \end{cases}
\end{equation}
and, for some $C>0$,
\begin{equation}
\label{W2}
|\nabla W_{R,h}(x)| < \frac{C}{h}\,,
\end{equation}
\begin{equation}
\label{W3}
|\nabla W_{R,h}(x)-\nabla W_{R,h}(x')| < \frac{C}{h^2}\,|x-x'|\,. 
\end{equation}

We define the quantity
\begin{equation}
\label{mass 1}
\mu_t(R,h) = \int\! \rmd x \, \big[1-W_{R,h}(x-B_\epsilon(t))\big]\, \omega_\epsilon (x,t)\,,
\end{equation}
which is a mollified version of $m_\epsilon$, satisfying
\begin{equation}
\label{2mass 3}
\mu_t(R,h) \le m_\epsilon(R, t) \le \mu_t(R-h,h)\,.
\end{equation}
In particular, it is enough to prove the claim with $\mu_t$ instead of $m_\epsilon$. 

\noindent The convenience is that the function $\mu_t$ is differentiable (with respect to $t$); therefore we compute its derivative:
$$
\begin{aligned} \frac{\mathrm{d}}{\mathrm{d}t} \mu_t(R,h) = & -
\int dx \: \nabla W_{R,h}(x-B_\epsilon(t)) \cdot [ u(x,t)+F_\epsilon(x,t)-\dot{B}_\epsilon(t) ] \, \omega_\epsilon(x,t) \\ = & -H_3- H_4 - H_5 
\end{aligned}
$$
with
$$
\begin{aligned}
 H_3= & \int dx \: \nabla W_{R, h}(x-B_\epsilon(t)) \cdot \widetilde{u}(x,t) \, \omega_\epsilon(x,t) \\
 H_4= & \int dx \: \nabla W_{R, h}(x- B_\epsilon(t)) \cdot \left[F_{\epsilon, 1}(x,t) - 
\int dy\,  F_{\epsilon, 1}(y,t) \, \omega_\epsilon(y,t)\right] \, \omega_\epsilon(x,t) \\
 H_5= & \int dx \: \nabla W_{R, h}(x-B_\epsilon(t))\\
 & \quad \quad  \cdot \bigg[u(x,t)- \widetilde{u}(x,t) -\int dy \, [u(y,t)- \widetilde{u}(y,t)] \, \omega_\epsilon(y,t) \bigg] \, \omega_\epsilon(x,t) \\
+ & \int dx \: \nabla W_{R, h}(x-B_\epsilon(t)) \cdot \left[F_{\epsilon, 2}(x,t) -
\int dy \, F_{\epsilon, 2} (y, t) \, \omega_\epsilon (y, t) \right] \, \omega_\epsilon(x, t) 
\end{aligned}
$$
since $\dot{B}_\epsilon(t) = \int dy \, \omega_\epsilon(y,t) [ F_\epsilon (y,t)+ u(y,t)- 
\widetilde{u}(y,t)] $. We immediately observe that, thanks to Proposition \ref{u e utilde}, to Lemma \ref{comesonoFeps}
and to the fact that $\nabla W_{R, h}(z)$ is zero if $|z| \leq R$, 
\begin{equation} \label{H5} |H_5| \leq \frac{C}{h} \cdot \frac{C}{|\log\epsilon|^{\alpha-1}} \cdot m_\epsilon(R,t) \, . \end{equation}
Following the proof of \cite[Proposition~3.4]{butmar2}  we find (we postpone the estimates
of $|H_3|$ and $|H_4|$ in the Appendix):
\begin{equation} \label{H3} |H_3| \leq \frac{C}{R\, h^3} I_\epsilon(t) \, m_\epsilon(R,t) \, ; 
\end{equation}
\begin{equation} \label{H4} |H_4| \leq C  \left( 1 + \frac{2R}{h}  \right)  m_\epsilon(R,t)
+  \frac{C \, I_\epsilon(t)}{R^2 \, h} m_\epsilon(R,t) \, .
\end{equation}

Recalling \eqref{Iepsilon}, from estimates \eqref{H5}, \eqref{H3},  and \eqref{H4},  we have:
\begin{equation}
 \frac{\mathrm{d}}{\mathrm{d}t} \mu_t (R,h) \leq A_\epsilon(R, h) m_\epsilon(R,t)  
\label{equ_mm}
\end{equation}
for any $t\leq \bar T$, where
$$
A_\epsilon(R, h)=
C
\left(\frac{1}{R \, h^3 \, |\log\epsilon|^{2(\alpha-1)}} + \frac{1}{R^2 \, h \, |\log\epsilon|^{2(\alpha-1)}}+ 
\frac{1}{h |\log\epsilon|^{\alpha-1}} + \frac{2R}{h} + 1 \right) .
$$
Therefore, by \eqref{2mass 3} and \eqref{equ_mm},
\begin{equation}
\mu_t (R,h) \leq \mu_0 (R,h) + A_\epsilon(R, h) \int_0^t {\mathrm{d}} s  \, \mu_s (R-h, h)
\end{equation}
for any $t\leq \bar T$ and $\epsilon$ sufficiently small.
We iterate the last inequality $n=\lfloor|\log\epsilon|\rfloor$  times (where $\lfloor a\rfloor$
denotes the integer part of the positive number $a$), from
$$
R_0 =\frac{1}{|\log\epsilon|^{k}} \qquad \textnormal{to} \qquad R_n=\frac{1}{2|\log\epsilon|^{k}} \, ,
$$
where $R_n = R_0 -n h$,  and consequently
$$
h= \frac{1}{2  n |\log\epsilon|^{k}} \, .
$$
In this range for $R$ the quantity $A_\epsilon(R, h)$ is bounded by $C |\log\epsilon |$, in fact
\begin{equation}
\frac{2 R}{h} \leq C \frac{|\log\epsilon|^{k +1}}{|\log\epsilon|^{k}} \, ,
\end{equation}
\begin{equation}
\frac{1}{h |\log\epsilon|^{\alpha -1}} \leq C \frac{|\log\epsilon|^{k +1}}{|\log\epsilon|^{\alpha -1}} \leq C |\log\epsilon|
\end{equation}
since ${k}=(\alpha-2)/2 < \alpha -1$,
\begin{equation}
\frac{1}{R^2 \, h \, |\log\epsilon|^{2(\alpha-1)}}
\leq C \frac{|\log\epsilon|^{3{k}+1}}{|\log\epsilon|^{2(\alpha -1)}} \leq C |\log\epsilon|
\end{equation}
since $3{k} < 2(\alpha -1)$,
\begin{equation}
\frac{1}{R \, h^3 \, |\log\epsilon|^{2(\alpha-1)}} \leq C
\frac{|\log\epsilon|^{4{k}+3}}{|\log\epsilon|^{2(\alpha -1)}} \leq C |\log\epsilon|
\label{qui}
\end{equation}
since $4{k}+2= 2(\alpha -1)$. Note that in this point we need to choose $k=(\alpha-2)/2$
and $\alpha >2$.

\noindent Then, for any $t\in [0, \bar T]$,   it results  $A_\epsilon(R, h) \leq C |\log\epsilon |$  and
 
\[
\begin{split}
\mu_t(R_0-h,h) & \le \mu_0(R_0-h,h) + \sum_{j=1}^{n-1} \mu_0(R_j,h) \frac{(C |\log\epsilon| t)^j}{j!} \\ & \quad + \frac{(C |\log\epsilon| )^{n}}{(n-1)!} \int_0^t\!{\textnormal{d}} s\,  (t-s)^{n-1}\mu_s(R_{n},h) \,.
\end{split}
\]
Since $\Lambda_\epsilon(0) \subset \Sigma(z|\epsilon)$, we can determine $\epsilon$
so small such that $\mu_0(R_j,h)=0$ for any $j=0,\ldots,n$, so that, for any $t\in [0,\bar T]$,
\begin{equation}
\label{mass 15'}
\mu_t(R_0-h,h) \le \frac{(C |\log\epsilon|)^{n}}{(n-1)!} \int_0^t\!{\textnormal{d}} s\,  (t-s)^{n-1}\mu_s(R_{n},h) \le  \frac{(C |\log\epsilon| t)^{n}}{n!}\,,
\end{equation}
where the obvious estimate $\mu_s(R_{n},h) \le 1$ has been used in the last inequality. In conclusion, using also \eqref{2mass 3}, 
\[
m_\epsilon(R_0, t) \le \mu_t(R_0 -h,h) \le  \left(C t\right)^{\lfloor|\log\epsilon|\rfloor} \;\; \forall\, t\in [0,\bar T] ,
\]
which implies  \eqref{lemma_massa} for $t\leq T$ and  $T$ suitably small.
\end{proof}

We are now ready to prove Theorem \ref{teoEA}

\begin{proof}[Proof of Theorem \ref{teoEA}]
With the previous results we can prove now that, for all $t\in[0,T]$,
\begin{equation}
\label{lamb}
\Lambda_\eps(t)\subseteq \Sigma \left(B_\eps (t) | \, C |\log\eps|^{-k} \right)\, .
\end{equation}
Recalling  the definiton of $R_t$ given in \eqref{max_lambda}, if at time $t\in[0,T]$
we have, for a certain $x_0\in \Lambda_\eps(0)$,
\begin{equation}
\frac34 R_t \leq |\phi_t(x_0)- B_{\eps}(t)|\leq R_t
\label{corona}
\end{equation}
then the time derivative of $|\phi_t(x_0)-B_{\eps}(t)|$ is bounded by \eqref{eq:variazonemax}, 
that is (considering also \eqref{Iepsilon}) 
\begin{equation}
\begin{split}
\frac{\mathrm{d}}{\mathrm{d}t}& |\phi_t(x_0)-B_\epsilon(t)|   \leq \\
&2L\, R_t +  \frac{C}{R_t^{3}|\log\epsilon|^{2(\alpha-1)}}  + C \sqrt{\epsilon^{-\gamma}m_\epsilon(R_t/2, t)} + \frac{C}{|\log\epsilon|^{\alpha-1}}
\end{split}
\label{max_sup}
\end{equation}
for each index $i$ of the $N$ vortices, omitted to simplify the notation.

\noindent Let $t_0$ be the first time at which  ${R_{t_0} = |\log\eps|^{-k}}$  and   ${R_{t} \geq |\log\eps|^{-k}}$   for $t\geq t_0$;
of course if such $t_0$ does not exist \eqref{lamb} is already achieved, as well as  if there are
time intervals after $t_0$ for which ${R_{t} \leq |\log\eps|^{-k}}$, \eqref{lamb} is achieved
in these time intervals.

\noindent In the worst case in which a fluid particle fulfills \eqref{corona}
in the whole time interval $[t_0, T]$, then we can bound the right hand side of
\eqref{max_sup}  in the following way,
\begin{equation}
\frac{\mathrm{d}}{\mathrm{d}t} |\phi_t(x_0)-B_\epsilon(t)|   \leq 
C |\phi_t(x_0)-B_\epsilon(t)|  \, ,
\label{eq_phi}
\end{equation}
since the other terms are negligible with respect to the first one, by 
Lemma \ref{lem:mt} 
(which holds for $R_t\geq |\log\epsilon|^{-k}$) with $\ell>\gamma$  and by the following
$$
\frac{C}{R_t^{3}|\log\epsilon|^{2(\alpha-1)}} \leq C R_t
\qquad
\iff
\qquad
R_t \geq \frac{C}{|\log\epsilon|^\frac{\alpha -1}{2}} 
$$

$$
\frac{C}{|\log\epsilon|^{\alpha-1}}\leq C R_t  \qquad
\iff
\qquad
R_t \geq \frac{C}{|\log\epsilon|^{\alpha -1}} \, ,
$$
which hold true since $R_t\geq |\log\epsilon|^{-k}$ and $k=(\alpha-2)/2$.
From \eqref{eq_phi} we immediatly obtain 
\begin{equation}
|\phi_t(x_0)-B_\epsilon(t)| \leq |\phi_{t_0}(x_0)-B_\epsilon(t_0)| {\rm{e}}^{C (t-t_0)}
\leq  |\log\eps|^{-k} {\rm{e}}^{C (t-t_0)}
\end{equation}
hence $|\phi_t(x_0)-B_\epsilon(t)| \leq C |\log\eps|^{-k}$  for any fluid particle satisfying \eqref{corona}  
$\forall \, t\in [t_0,T]$.
We  can now deduce the same bound for any fluid particle in $\Lambda_\eps(t)$, that is
\begin{equation} \label{suppRt} 
\Lambda_{\epsilon}(t)  \subseteq \Sigma(B_\epsilon(t)| C |\log\eps|^{-k}) \, , 
\end{equation}
or equivalently $R_t \leq C |\log\eps|^{-k}$.
Suppose in fact  that for a  fluid particle, which does not satisfy  \eqref{corona} for $t=t_0$,
the quantity $|\phi_t(x_0)- B_{\eps}(t)|$ reaches at a certain
$t^*>t_0$  the value $3R_{t^*}/4$;  at this time $t^*$  the quantity $|\phi_{t^*}(x_0)- B_{\eps}(t^*)|$
is clearly bounded by the analogous quantity for a fluid particle which satisfies \eqref{corona}
for any $t \in [t_0, T]$.  If the first fluid particle for $t\geq t^*$ enters  the region  \eqref{corona}, then  $|\phi_t(x_0)- B_{\eps}(t)|$  has a 
time derivative which is bounded by \eqref{eq_phi}, hence its
successive growth is controlled by $C |\log\eps|^{-k}$.

\noindent If there is not a single fluid particle which satisfies \eqref{corona} for any $t\in[t_0, T]$,
the same argument can be applied in subintervals $[t_0, t_1]$, $[t_1, t_2], \dots 
,[t_n, T]$, in any of which certainly there is a fluid particle for which \eqref{corona} holds.
Therefore for \textit{any} fluid particle in $\Lambda_\eps(t)$   we have  $|\phi_t(x_0)- B_{\eps}(t)|\leq C |\log\eps|^{-k}$,
that is  $R_t\leq C |\log\eps|^{-k}$. Thus \eqref{lamb} is proved.

It remains to prove that
\begin{equation}
\label{B-z}
|B^i_\epsilon(t)-z_i(t)|\leq \frac{C}{|\log\epsilon|^{k}} \quad  {\textnormal{for each}}\,\, i \, ,
\end{equation}
since with this bound, together with \eqref{lamb},  Theorem \ref{teoEA} follows.
In this last point we reintroduce the index $i$.

\noindent We have:
$$
\begin{aligned}
\dot{B}^i_{\epsilon}(t) - \dot{z}_i(t)  = & \: a_i^{-1} \int dx \, \big(u^i(x, t)+ F^i_\epsilon(x,t) \big) \, \omega_{i, \epsilon}(x,t) \\
& \: - \sum_{j \neq i} a_j K(z_i(t) - z_j(t))\, .
\end{aligned}
$$
We add and subtract appropriate terms, then by the the splitting  of $ F^i_\epsilon$ given
in Lemma \ref{comesonoFeps}  we get:
$$
\begin{aligned}
\dot{B}_{\epsilon}^i(t) - \dot{z}_i(t)= & \:a_i^{-1} \int dx \, u^i(x,t) \, \omega_{i, \epsilon}(x, t) 
+ a_i^{-1}  \int dx \, F^i_{\epsilon,2} (x, t) \, \omega_{i, \epsilon}(x, t) \\
& \: + a_i^{-1} \int dx [F^i_{\epsilon,1}(x,t) - F^i_{\epsilon,1}(B^i_{\epsilon}(t), t) ] \, \omega_{i, \epsilon} (x, t) \\
& \: + \sum_{j \neq i} \int dy \, [K(B^i_{\epsilon}(t)-y)- K(B^i_{\epsilon}(t)- B^j_{\epsilon}(t)) ] \, \omega_{j, \epsilon} (y, t) \\
& \: + \sum_{j \neq i} a_j [K(B^i_{\epsilon}(t)- B^j_{\epsilon}(t)) -K(B^i_{\epsilon}(t) -z_j(t) ) ] \\
& \: + \sum_{j \neq i} a_j [ K(B^i_{\epsilon}(t) -z_j(t) ) - K(z_i(t) - z_j(t)) ] \, .
\end{aligned}
$$
The first term on the right hand side is controlled by adding and subtracting 
$\widetilde u^i(x,t)$ inside the integral, using Proposition \ref{u e utilde} and
\eqref{antisym}. For the other terms we use Lemma \ref{comesonoFeps} and the Lipschitz property of $K$ outside the disk 
$\Sigma \big(0 \big| R_{\mathrm{min}}/2 \big)$ (we call $L_1$ the Lipschitz constant of $K$ in this region) obtaining
$$
\begin{aligned}
| \dot{B}^i_{\epsilon}(t) - \dot{z}_i(t) | \leq & \:  |a_i|^{-1} \frac{C}{|\log \epsilon|^{\alpha-1}}  + |a_i|^{-1} L \, I^i_{\epsilon}(t)^{1/2} \, |a_i|^{1/2} \\
& + \sum_{j \neq i} L_1 I^j_{\epsilon}(t)^{1/2} \, |a_j|^{1/2}
 + \sum_{j \neq i} |a_j| \, L_1 |B^j_{\epsilon}(t) -z_j(t) | 
 \\ & \: 
+ \sum_{j \neq i} |a_j| \, L_1 |B^i_{\epsilon}(t) -z_i(t) | \, .
\end{aligned}
$$
Defining $\Delta(t) := \max_{i=1, \ldots, N} |B^i_{\epsilon}(t) -z_i(t) |$, then
$$
\begin{aligned} \dot{\Delta}(t) \leq & \max_{i=1, \ldots, N} | \dot{B}^i_{\epsilon}(t) - \dot{z}_i(t) | \\
\leq & \, \frac{C}{|\log \epsilon|^{\alpha-1}} + C \sum_{j=1}^N \sqrt{I^j_{\epsilon}(t)} + 2L_1 \sum_{j=1}^N |a_j| \: \Delta(t) \, .
\end{aligned}
$$
By definition of $F^i_{\epsilon, 1}$ it can be immediatly seen that $L \geq L_1 \sum_j |a_j|$. 
By integration of the previous inequality we obtain:
\[ \Delta(t) \leq  \Delta(0) {\textnormal{e}}^{2Lt} + C \int_0^t \sum_{j=1}^N \sqrt{I^j_\epsilon(s)} {\textnormal{e}}^{L(t-s)} \, ds + \frac{C}{|\log \epsilon|^{\alpha-1}} \cdot ({\textnormal{e}}^{2Lt}-1) \, . \]
Using the bound \eqref{sqrtIepsilon} we get, for $t\leq T$,
\begin{equation} \label{Deltat} \Delta(t) \leq \frac{C}{|\log \epsilon|^{\alpha-1}} \,  ,
\end{equation}
so \eqref{B-z} is achieved.
This concludes the proof of Theorem \ref{teoEA}.    
\end{proof}

\section*{Appendix}
In this appendix we derive estimates \eqref{H3} and \eqref{H4}.
\begin{equation*}
\begin{split}
H_3 & = \int\! d x\, \nabla W_{R,h}(x-B_\eps(t)) \cdot \int\! d y \, K(x-y)\, \omega_\eps(y,t)\, \omega_\eps(x,t)  \\ & = \frac 12 \int\! d x \! \int\! d y\, [\nabla W_{R,h}(x-B_\eps(t)) - \nabla W_{R,h}(y-B_\eps(t))] \\ & \qquad  \cdot K(x-y) \, \omega_\eps(x,t)\,  \omega_\eps(y,t) \\ 
H_4 & =   \int\! d x\, \nabla W_{R,h}(x-B_\eps(t)) \cdot \int\! d y \,[F_{\eps, 1}(x,t)-F_{\eps, 1}(y,t)]\, \omega_\eps(y,t)\, \omega_\eps(x,t)\,, 
\end{split}
\end{equation*}
where the second expression of $H_3$ is due to the antisymmetry of $K$.

Concerning $H_3$, we introduce the new variables $x'=x-B_\eps(t)$, $y'=y-B_\eps(t)$, define $\tilde\omega_\eps(z,t) := \omega_\eps(z+B_\eps(t),t)$, and let
\[
f(x',y') = \frac 12 \tilde\omega_\eps(x',t)\, \tilde\omega_\eps(y',t) \, [\nabla W_{R,h}(x')-\nabla W_{R,h}(y')] \cdot K(x'-y') \,,
\]
so that $H_3 = \int\! d x' \! \int\! d y'\,f(x',y')$. We observe that $f(x',y')$ is a symmetric function of $x'$ and $y'$ and that, by \eqref{W1}, a necessary condition to be different from zero is if either $|x'|\ge R$ or $|y'|\ge R$. Therefore, 
\begin{equation*}
\begin{split}
H_3  &= \bigg[ \int_{|x'| > R}\! d x' \! \int\! d y' + \int\! d x' \! \int_{|y'| > R}\! d y' -  \int_{|x'| > h}\! d x' \! \int_{|y'| > R}\! d y'\bigg]f(x',y') \\ & = 2 \int_{|x'| > R}\! d x' \! \int\! d y'\,f(x',y')  -  \int_{|x'| > R}\! d x' \! \int_{|y'| > R}\! d y'\,f(x',y') \\ & = H_3' + H_3'' + H_3'''\,,
\end{split}
\end{equation*}
with 
\begin{equation*}
\begin{split}
H_3' & = 2 \int_{|x'| > R}\! d x' \! \int_{|y'| \le R-h}\! d y'\,f(x',y') \,, \\ H_3''&  = 2 \int_{|x'| > R}\! d x' \! \int_{|y'| > R-h}\! d y'\,f(x',y')\,, \\ H_3''' & = -  \int_{|x'| > R}\! d x' \! \int_{|y'| > R}\! d y'\,f(x',y')\,.
\end{split}
\end{equation*}
By the assumptions on $W_{R,h}$, we have $\nabla W_{R,h}(z) = \eta_h(|z|) z/|z|$ with $\eta_h(|z|) =0$ for $|z| \le R$. In particular, $\nabla W_{R,h}(y') = 0$ for $|y'| \le R-h$, hence
\[
H_3' =  \int_{|x'| > R}\! d x' \, \tilde\omega_\eps(x',t) \eta_h(|x'|) \,\frac{x'}{|x'|} \cdot  \int_{|y'| \le R-h}\! d y'\, K(x'-y') \, \tilde\omega_\eps(y',t)\,.
\]
In view of  \eqref{W2}, $|\eta_h(|z|)| \le C/h$, so that 
\begin{equation}
\label{a1'}
|H_3'| \le \frac{C}{h} m_\eps(R,t) \sup_{|x'| > R} |A_3(x')|\,,
\end{equation}
with
\[
A_3(x') = \frac{x'}{|x'|}\cdot  \int_{|y'| \le R-h}\! d y'\, K(x'-y') \, \tilde\omega_\eps(y',t) \,.
\]
Now, recalling \eqref{nucleoK} and using that $x'\cdot (x'-y')^\perp=-x'\cdot y'^\perp$, we get,
\begin{equation}
\label{in H_11}
A_3(x') = \frac{1}{2\pi} \int_{|y'|\leq R-h}\!  d y'\, \frac{x'\cdot y'^\perp}{|x'||x'-y'|^2}\,  \tilde\omega_\eps(y',t) \,.
\end{equation}
By \eqref{cen_vort}, $\int\! d y'\,  y'^\perp\,  \tilde\omega_\eps(y',t) = 0$, so that
\begin{equation}
\label{in H_13}
A_3(x')  = A_3'(x')-A_3''(x')\,, 
\end{equation}
where
\begin{eqnarray*}
	&& A_3'(x') = \frac{1}{2\pi}  \int_{|y'|\le R-h}\! d y'\, \frac {x'\cdot y'^\perp}{|x'|}\, \frac {y'\cdot (2x'-y')}{|x'-y'|^2 \ |x'|^2} \,  \tilde\omega_\eps(y',t) \,, \\ && A_3''(x')= \frac{1}{2\pi} \int_{|y'|> R-h}\! d y'\, \frac{x'\cdot y'^\perp}{|x'|^3}\,  \tilde\omega_\eps(y',t) \,.
\end{eqnarray*}
We notice that if $|x'| > R$ then $|y'| \le R-h$ implies $|x'-y'|\ge h$ and $|2x'-y'|\le |x'-y'|+|x'|$. Therefore, for any $|x'| > R \ge 2h$,
\[
\begin{split}
|A_3'(x')|& \le \frac{1}{2\pi}\bigg[\frac{1}{|x'|^2h} + \frac{1}{|x'|h^2} \bigg]  \int_{|y'|\leq R-h} \! d y'\, |y'|^2 \,  \tilde\omega_\eps(y',t) \\ & \le \frac{I_\eps(t)}{2\pi}\bigg[\frac{1}{R^2h} + \frac{1}{Rh^2}\bigg] \le \frac{3I_\eps(t)}{4\pi Rh^2}\,.
\end{split}
\]
To bound $A_3''(x')$, by Chebyshev's inequality, for any $|x'| > R \ge 2h$ we have,
\[
|A_3''(x)| \le \frac{1}{2\pi |x'|^2} \int_{|y'|> R-h}\! d y'\, |y'| \tilde\omega_\eps(y',t) \le \frac{I_\eps(t)}{2\pi R^2(R-h)} \le \frac{I_\eps(t)}{2 \pi R^2h} \,.
\]
From Eqs.~\eqref{a1'} and \eqref{in H_13}, the previous estimates, and $R\ge 2h$, we conclude that
\begin{equation}
\label{H_14b}
|H_3'| \le \frac{5C \, I_\eps(t)}{4\pi Rh^3} m_\eps(R,t)\,.
\end{equation}

Now, by \eqref{W3} and then applying the Chebyshev's inequality and again $R\ge 2h$,
\begin{equation*}
\begin{split}
|H_3''| + |H_3'''| & \le \frac{C}{\pi h^2} \int_{|x'| \ge R}\! d x' \! \int_{|y'| \ge R-h}\! d y'\,\tilde\omega_\eps(y',t) \,  \tilde\omega_\eps(x',t)  \\ & = \frac{C}{\pi h^2}m_\eps(R,t)   \int_{|y'| \ge R-h}\! d y'\, \tilde\omega_\eps(y',t)  \le \frac{4C \, I_\eps(t)}{\pi R^2h^2} m_\eps(R,t)\,.
\end{split}
\end{equation*}
In conclusion, recalling $R\ge 2h$, 
\begin{equation}
\label{a1s}
|H_3| \le  \frac{13 C \, I_\eps(t)}{4\pi R h^3} m_\eps(R,t)\,.
\end{equation}

Concerning $H_4$, we observe that by \eqref{W1} the integrand is different from zero only if $R\le |x-B_\eps(t)|\le R+h$. Therefore, by Lemma \ref{comesonoFeps} and \eqref{W2} we have,  using again the variables $x'=x-B_\eps(t)$, $y'=y-B_\eps(t)$,
\[
\begin{split}
|H_4| & \le \frac{C}{h} \int_{|x'|\ge R}\! d x'  \tilde\omega_\eps(x',t)  \int_{|y'|> R}\! d y'\, \tilde\omega_\eps(y',t) \\ &\quad + \frac{C }{h} \int_{R \le |x'|\le R+h}\! d x'  \tilde\omega_\eps(x',t) \int_{|y'| \le R}\! d y'\,|x'- y'| \,  \tilde\omega_\eps(y',t) \,.
\end{split}
\]
Since $|x'-y'| \le 2R+h$ in the domain of integration of the last integral and using the Chebyshev's inequality in the first one we get,
\begin{equation}
\label{a2s}
|H_4| \le \frac{C \, I_\eps(t)}{R^2h} m_\eps(R,t) +  C  \bigg(1+\frac{2R}h\bigg) m_\eps(R,t)\,.
\end{equation}

\bigskip

\bigskip

\bigskip

\bigskip
\noindent {\textbf{Acknowledgments}}.
Work performed under the auspices of GNFM-INDAM and the Italian Ministry
of the University (MIUR).

\bigskip
\bigskip
\begin{minipage}[t]{10cm}
\begin{flushleft}
\small{
\textsc{Guido Cavallaro}
\\*Sapienza Universit\`a di Roma,
\\*Dipartimento di Matematica
\\*Piazzale Aldo Moro, 2
\\* Roma, 00185, Italia
\\*e-mail: cavallar@mat.uniroma1.it
\\[0.4cm]
\textsc{Carlo Marchioro}
\\*International Research Center M\&MOCS,
\\*Universit\`a di L'Aquila
\\*Palazzo Caetani
\\* Cisterna di Latina (LT), 04012, Italia
\\*e-mail: marchior@mat.uniroma1.it
}
\end{flushleft}
\end{minipage}


\end{document}